\documentclass[twocolumn]{autart}
\usepackage[utf8]{inputenc}
\DeclareUnicodeCharacter{221E}{$\infty$}
\usepackage{caption}
\usepackage{comment}
\usepackage{graphicx}
\usepackage{amsmath}
\usepackage{enumitem}

\let\theoremstyle\relax
\usepackage{amsthm,amssymb}
\usepackage[version=4]{mhchem}
\usepackage{siunitx}
\usepackage{longtable,tabularx}
\usepackage{tikz}
\usepackage{pgfplots}
\pgfplotsset{compat=1.15}
\usepackage{subfigure}
\usepackage{standalone}
\usepackage[sort & compress,numbers]{natbib}
\newtheorem{theorem}{Theorem}[section]
\newtheorem{remark}{Remark}[section]
\newtheorem{lemma}{Lemma}[section]
\begin{document}
\begin{frontmatter}
\title{Exact Solution for the Rank-One Structured Singular Value with Repeated Complex Full-Block Uncertainty}
\author[UMN]{Talha Mushtaq}\ead{musht002@umn.edu},    
\author[UMich]{Peter Seiler}\ead{pseiler@umich.edu},  
\author[UMN]{Maziar S. Hemati}\ead{mhemati@umn.edu}  
\address[UMN]{Aerospace Engineering and Mechanics,
University of Minnesota,
Minneapolis, MN 55455, USA}  
\address[UMich]{Electrical Engineering and Computer Science,
 University of Michigan, Ann Arbor, MI 48109, USA}    
\begin{abstract}
In this note, we present an exact solution for the structured singular value (SSV) of rank-one complex matrices with repeated complex full-block uncertainty.  
A key step in the proof is the use of Von Neumman's trace inequality. 
Previous works provided exact solutions for rank-one SSV when the uncertainty contains repeated (real or complex) scalars and/or non-repeated complex full-block uncertainties.
Our result with repeated complex full-blocks contains, as special cases, the previous results for repeated complex scalars and/or non-repeated complex full-block uncertainties. 
The repeated complex full-block uncertainty has recently gained attention in the context of incompressible fluid flows.  
Specifically, it has been used to analyze the effect of the convective nonlinearity in the incompressible Navier-Stokes equation (NSE). 
SSV analysis with repeated full-block uncertainty has led to an improved understanding of the underlying flow physics. 
We demonstrate our method on a turbulent channel flow model as an example.
\end{abstract}
\end{frontmatter}
\section{Introduction}
This paper focuses on the computation of the structured singular value (SSV) given a feedback-interconnection between a rank-one complex matrix and a block-structured uncertainty.
The rank-one SSV is well-studied with some prominent results given in \cite{young1994rank, chen1994structuredpart1,chen1994structuredpart2}.
A standard SSV upper-bound can be formulated as a convex optimization \cite{packard1993complex}. 
This SSV upper-bound is equal to the true SSV for rank-one matrices when the uncertainty consists of repeated (real or complex) scalar blocks and non-repeated, complex full-blocks.
This yields an explicit expression for the rank-one SSV with these uncertainty structures (see Theorem 1 and 2 in \cite{young1994rank}). 
Similar results are given in \cite{chen1994structuredpart1,chen1994structuredpart2, fan2006robustness}. 

Our paper builds on this previous literature by providing an explicit solution to the rank-one SSV problem with repeated complex full-block uncertainty. 
This explicit solution is the main result and is stated as Theorem \ref{thm:1} in the paper. 
A key step in the proof is the use of Von Neumann's trace inequality \cite{von1962some}.
The repeated complex full-block uncertainty structure contains, as special cases, repeated complex scalar blocks and non-repeated, complex full-blocks.
Hence our explicit solution encompasses prior results for these cases.  

The repeated complex full-block uncertainty structure has physical relevance in systems such as fluid flows. 
Specifically, this uncertainty structure has recently been used to provide consistent modeling of the nonlinear dynamics \cite{liu21, liu2022strat, Bhattacharjee_et_al_2023, mushtaq2023structured, mushtaq2023riblets}. 
In Section 4, we demonstrate our rank-one solution to analyze a turbulent channel flow model \cite{mckeon2017engine}.
Our explicit rank-one solution is compared against existing SSV upper and lower bound algorithms \cite{Mushtaq2023Algorithm} that were developed for general (not-necessarily rank-one) systems.
\section{Background: Structured Singular Value}
Consider the standard SSV problem for square\footnote{We present the square complex matrix case to improve readability of the paper and minimize notation. The general rectangular complex matrix case can be handled by introducing some additional notation.} complex matrices $M \in \mathbb{C}^{m \times m}$ given by the function $\mu : \mathbb{C}^{m \times m} \to \mathbb{R}$ as
\cite{packard1993complex}
\begin{equation}
    \mu(M) = \left(\min{\|\Delta\|}: \mathrm{det}\left(I_m - M\Delta\right) = 0\right)^{-1}
    \label{eq:ssv}
\end{equation}
where $\Delta \in \mathbb{C}^{m \times m}$ is the structured uncertainty, $I_m$ is an $m \times m$ identity, $\mathrm{det}(\cdot)$ is the determinant and ${\|\cdot\|}$ is the induced 2-norm which is equal to the maximum singular value.
Then, $\mu(M)$ is the SSV of $M$.
For the trivial case where $M = 0$, the minimization in \eqref{eq:ssv} has no feasible point and $\mu(0) = 0$.
%
%
In this paper, we will focus on the case where $M$ is rank-one, i.e., $M = u v^{\mathrm{H}}$ for some ${u, v \in \mathbb{C}^m}$.
Then, using the matrix determinant lemma, the minimization problem in \eqref{eq:ssv} can be equivalently written as \cite{young1994rank, chen1994structuredpart1}
\begin{equation}
\mu (M) = \left(\min{\|\Delta\|}: v^{\mathrm{H}} \Delta u = 1\right)^{-1}.
\label{eq:ssv_rank_1}
\end{equation}
Hence, for any structured $\Delta$, the determinant constraint in \eqref{eq:ssv} can be converted into an equivalent scalar constraint when $M$ is rank-one. 
This scalar constraint is a special case of \emph{affine parameter variation} problem for polynomials with perturbed coefficients \cite{young1994rank, qiu1989simple}.
%
%
We will present solution for \eqref{eq:ssv_rank_1} when $\Delta \in \mathbf{\Delta}$, where $\mathbf{\Delta}$ is a set of repeated complex full-block uncertainties defined as
\begin{align}
    \begin{split}
\label{eq:RFB_1}
\mathbf{\Delta} := & \left \{ \Delta = \text{diag}(I_{r_1} \otimes \Delta_1, \ldots, I_{r_{n}} \otimes \Delta_{n}) \, \right. \\
  & \left. : \, \Delta_i \in \mathbb{C}^{k_i \times k_i} \right \}  \subset \mathbb{C}^{m \times m}.
    \end{split}
\end{align}
This set is comprised of $n$ blocks such that the $i^{\mathrm{th}}$ block, i.e., $I_{r_i} \otimes \Delta_i$, corresponds to a full $k_i \times k_i$ matrix repeated $r_i$ times.
%
%
Any uncertainty $\Delta \in \mathbf{\Delta}$ reduces to the complex uncertainties commonly found in the SSV literature:
%
\begin{enumerate}
    \item When $k_i = 1$ then $\Delta_i$ is a scalar, denoted as $\delta_i$. 
    In this case, the $i^{\mathrm{th}}$ block in \eqref{eq:RFB_1} corresponds to a repeated complex scalar, i.e., $I_{r_i} \otimes \Delta_i = \delta_i I_{r_i}$,
    \item When $r_i = 1$ then the $i^{\mathrm{th}}$ block in \eqref{eq:RFB_1} corresponds to a (non-repeated) complex full-block, i.e., ${I_{r_i} \otimes \Delta_i = \Delta_i}$.
\end{enumerate}
Explicit rank-one solutions of $\mu(M)$ for these special cases are well-known \cite{chen1994structuredpart1, young1994rank}.
However, the current SSV literature does not present any explicit rank-one solutions of $\mu(M)$ for the repeated complex full-block case, which is a more general set of complex uncertainties, i.e., for any $\Delta \in \mathbf{\Delta}$.
These uncertainty structures have physical importance in engineering systems such as fluid flows \cite{liu21,liu2022strat, mushtaq2023structured, Bhattacharjee_et_al_2023}, where they have been exploited to provide physically consistent approximations of the convective nonlinearity in the Navier-Stokes equations (NSE).
Therefore, in the next section, we will present an explicit rank-one solution of $\mu(M)$ for any $\Delta \in \mathbf{\Delta}$.
%
%
%
%
%
It is important to note that the solutions presented in this paper are not limited to fluid problems and can be used for any other system that has $\Delta \in \mathbf{\Delta}$.
\section{Repeated Complex Full-Block Uncertainty (Main Result)}
\label{sec:3}
%
%
%
%
Consider the problem in \eqref{eq:ssv_rank_1} for any $\Delta \in \mathbf{\Delta}$.
We can partition $u, v \in \mathbb{C}^m$ compatibly with the $n$ blocks of $\Delta \in \mathbf{\Delta}$:
\begin{equation}
    u = \begin{bmatrix}
        u_1^{\mathrm{H}} & \ldots & u_n^{\mathrm{H}}
    \end{bmatrix}^{\mathrm{H}}, \, v = \begin{bmatrix}
        v_1^{\mathrm{H}} & \ldots & v_n^{\mathrm{H}}
    \end{bmatrix}^{\mathrm{H}}
\end{equation}
where $u_i, v_i \in \mathbb{C}^{k_i r_i}$.
Note that $m = \sum_{i = 1}^{n} r_i k_i$.
Since, the $i^{\mathrm{th}}$ block is $I_{r_i} \otimes \Delta_i$, we can further partition $u_i, v_i$ based on the repeated structure:
\begin{equation}
    u_i = \begin{bmatrix}
        u_{i,1}^{\mathrm{H}} & \ldots & u_{i,r_i}^{\mathrm{H}}
    \end{bmatrix}^{\mathrm{H}}, \, v_i = \begin{bmatrix}
        v_{i,1}^{\mathrm{H}} & \ldots & v_{i,r_i}^{\mathrm{H}}
    \end{bmatrix}^{\mathrm{H}}
\end{equation}
where each $u_{i,j}, v_{i,j} \in \mathbb{C}^{k_i}$.
Based on this partitioning, define the following matrices (for $i = 1, \ldots, n$):
\begin{equation}
    Z_i = \sum_{j = 1}^{r_i} u_{i,j} v_{i,j}^{\mathrm{H}} \in \mathbb{C}^{k_i \times k_i}.
    \label{eq:Z_i}
\end{equation}
\begin{lemma}
    Let $M = uv^\mathrm{H}$ be given with $u,v \in \mathbb{C}^m$ and define $Z_i$ as in \eqref{eq:Z_i}. Then, for any $\Delta \in \mathbf{\Delta}$, we have
    \begin{equation}
        \mathrm{det}\left( I_m - M \Delta \right) = 1 - \sum_{i = 1}^n \mathrm{Tr}\left( Z_i \Delta_i \right).
    \end{equation}
    \label{lemma:1}
\end{lemma}
\begin{proof}
    Using the matrix determinant lemma, we have
    \begin{equation}
        \mathrm{det}\left(I_m - M\Delta \right) = 1 - v^{\mathrm{H}} \Delta u.
        \label{eq:1}
    \end{equation}
    Now, using the block-structure of $\Delta \in \mathbf{\Delta}$ and the corresponding partitioning of $\left(u,v \right)$,  we can rewrite \eqref{eq:1} as
    \begin{equation}
    \begin{aligned}
        1 - v^{\mathrm{H}} \Delta u & = 1 - \sum_{i = 1}^n   v_{i}^{\mathrm{H}} \left(I_{r_i} \otimes \Delta_i \right) u_{i} \\
        & = 1 - \sum_{i = 1}^n  \left[ \sum_{j = 1}^{r_i} v_{i,j}^{\mathrm{H}} \Delta_i u_{i,j} \right].
    \end{aligned}
        \label{eq:2}
    \end{equation}
    Note that the term in brackets is a scalar and hence equal to its trace.
    Thus, use the cyclic property of the trace as
    \begin{equation}
        \begin{aligned}
            \sum_{j = 1}^{r_i} \mathrm{Tr}\left[v_{i,j}^{\mathrm{H}} \Delta_i u_{i,j}\right] & = \sum_{j = 1}^{r_i} \mathrm{Tr}\left[u_{i,j} v_{i,j}^{\mathrm{H}} \Delta_i \right] \\
            & = \mathrm{Tr}\left[ Z_i \Delta_i\right].
        \end{aligned}
        \label{eq:3}
    \end{equation}
    Combine \eqref{eq:1}, \eqref{eq:2} and   \eqref{eq:3} to obtain the stated result.
\end{proof}
Lemma \ref{lemma:1} is used to provide an explicit solution for rank-one SSV with repeated complex full-blocks. 
This is stated next as Theorem \ref{thm:1}.
\vspace*{0.5 cm}
\begin{theorem}
\label{thm:1}
    Let $M = u v^{\mathrm{H}}$ be given with $u,v \in \mathbb{C}^m$ and define $Z_i$ as in \eqref{eq:Z_i}.
    Then,
    \begin{equation}
    \mu(M) = \sum_{i = 1}^n \sum_{j = 1}^{k_i} \sigma_{j}\left( Z_i \right),
    \label{eq:mu_M}
    \end{equation}
    where $\sigma_{j}\left( Z_i \right)$ is the $j^{\mathrm{th}}$ singular value of $Z_i$.
\end{theorem}
\begin{proof}
Define $c = \sum_{i = 1}^n \sum_{j = 1}^{k_i} \sigma_{j}\left( Z_i \right) $ to simplify notation.
The proof consists of 2 directions: $(i) \, \mu(M) \ge c$ and $(ii) \, \mu(M) \le c$.

$(i) \, \mu(M) \ge c:$ Let $Z_i = U_i \Sigma_i V_i^{\mathrm{H}}$ be the singular value decomposition (SVD) of $Z_i$.
Note that $\Sigma_i = \mathrm{diag}(\sigma_1(Z_i), \ldots, \sigma_{k_i}(Z_i))$.
Then, define $\overline{\Delta} \in \mathbf{\Delta}$ with the blocks $\overline{\Delta}_i = \frac{1}{c} V_i U_i^{\mathrm{H}}$ ($i=1,\ldots, n$).
Thus, by Lemma \ref{lemma:1}, we have 
\begin{equation}
    \mathrm{det}(I_m - M \overline{\Delta}) = 1 - \sum_{i = 1}^n \mathrm{Tr} \left[Z_i \overline{\Delta}_i \right].
    \label{eq:thm_1}
\end{equation}
Now, substitute the SVD of $Z_i$ in \eqref{eq:thm_1} and use the cyclic property of trace:
\begin{equation}
\begin{aligned}
    \mathrm{det}(I - M \overline{\Delta}) & = 1 - \sum_{i = 1}^n \mathrm{Tr} \left[\Sigma_i V_i^{\mathrm{H}} \overline{\Delta}_i U_i \right] \\
    & = 1 - \frac{1}{c} \sum_{i = 1}^n \mathrm{Tr} \left[\Sigma_i \right] = 0.
\end{aligned}
\end{equation}
Hence $\overline{\Delta}$ causes singularity and $\|\overline{\Delta}\|_2 = \frac{1}{c}$.
Thus, the minimum $\|\Delta\|$ in \eqref{eq:ssv_rank_1} must satisfy $\|\Delta\| \leq \frac{1}{c}$ and consequently, $\mu(M) \geq c$.

$(ii) \, \mu(M) \le c:$ \,\, Let $\Delta \in \mathbf{\Delta}$ be given with $\|\Delta\| < \frac{1}{c}$.
Von Neumann's trace inequality \cite{von1962some} gives:
\begin{equation}
    | \mathrm{Tr} \left[Z_i \Delta_i \right] | \leq \sum_{j = 1}^{k_i} \sigma_j(Z_i) \sigma_j(\Delta_i) 
    \label{eq:ineq}
\end{equation}
where $| \cdot |$ is the absolute value.
Note that $\|\Delta\| < \frac{1}{c}$ implies that each block satisfies the same bound: ${\sigma_j(\Delta_i) < \frac{1}{c}}$.
Hence, \eqref{eq:ineq} implies
\begin{equation}
    \begin{aligned}
        \left| \mathrm{Tr}[ Z_i \Delta_i ] \right| < \frac{1}{c} \sum_{j = 1}^{k_i} \sigma_j(Z_i).
    \end{aligned}
    \label{eq:ineq_2}
\end{equation}
Next, using Lemma \ref{lemma:1} and the inequality in \eqref{eq:ineq_2}, we get
\begin{equation}
    \begin{aligned}
        \mathrm{det}(I_m - M \Delta) & = 1 - \sum_{i = 1}^n \mathrm{Tr}\left[ Z_i \Delta_i\right] \\
        & > 1 - \frac{1}{c} \sum_{i = 1}^n \left[\sum_{j = 1}^{k_i} \sigma_j(Z_i) \right] = 0.
    \end{aligned}
\end{equation}
Hence, any $\Delta \in \mathbf{\Delta}$ with $\|\Delta\| < \frac{1}{c}$ cannot cause ${(I_m - M \Delta)}$ to be singular.
Thus, the minimum $\|\Delta\|$ in \eqref{eq:ssv_rank_1} must satisfy $\|\Delta\| \geq \frac{1}{c}$ and consequently, $\mu(M) \leq c$.
\end{proof}
\begin{remark}
For the special cases $r_i = 1$ and $k_i = 1$, the solution \eqref{eq:mu_M} yields $\mu(M) = \sum_{i = 1}^{n} \|u_i\|_2 \|v_i\|_2$ and $\mu(M) = \sum_{i = 1}^n |v_i^{\mathrm{H}} u_i|$.
These special cases correspond to solutions presented in previous works for non-repeated, complex full-block and repeated complex scalar uncertainties, respectively \cite{chen1994structuredpart1, young1994rank}.
\end{remark}
\section{Results}
In this section, we demonstrate our SSV solution method for repeated complex full-blocks using a rank-one approximation of the turbulent channel flow model.
As validation, we will compare our solutions against general upper and lower-bound algorithms that have been developed for (not necessarily rank-one) systems with repeated complex full-block uncertainties.
The upper and lower-bounds are computed using \emph{Algorithm 1} (Upper-Bounds) and \emph{Algorithm 3} (Lower-Bounds) in \cite{Mushtaq2023Algorithm}, which are based on Method of Centers \cite{Boyd2004} and Power-Iteration \cite{packard1993complex}, respectively.
Generally, these algorithms can be used for higher rank problems (see for example \cite{mushtaq2023structured} and \cite{Bhattacharjee_et_al_2023}).
Additionally, we will compare the computational times between each of the methods to demonstrate the computational scaling of the rank-one SSV solution.

\subsection{Example}
The spatially-discretized turbulent channel flow model described in \cite{mckeon2017engine} has the following higher-order dynamical equation:
\begin{align}
    \begin{split}
        & E(\kappa_x, \kappa_z) \Dot{\phi}(y) = A(Re, \kappa_x, \kappa_z) \phi(y) + B(\kappa_x, \kappa_z) f(y) \\
        & \zeta(y) = C(\kappa_x, \kappa_z) \phi(y) \\
        & f(y) = \Delta \zeta(y)
    \end{split}
    \label{eq:dyn_eq}
\end{align}
where $Re$ is the Reynolds number, $\kappa_x$ and $\kappa_z$ are the streamwise ($x$) and spanwise ($z$) direction wavenumbers resulting from the discretization, and the wall-normal direction is given by $y$.
Here, the states $\phi(y) \in \mathbb{C}^{4N}$ and outputs $\zeta(y) \in \mathbb{C}^{9N}$ are given by the following:
\begin{align}
    \begin{split}
        & \phi(y) = [u(y)^{\mathrm{T}},v(y)^{\mathrm{T}},w(y)^{\mathrm{T}},p(y)^{\mathrm{T}}]^{\mathrm{T}}, \\ & \zeta(y) = [(\nabla u(y))^{\mathrm{T}},(\nabla v(y))^{\mathrm{T}},(\nabla w(y))^{\mathrm{T}}]^{\mathrm{T}}
    \end{split}
\end{align}
where $u(y) \in \mathbb{C}^{N}$, $v(y) \in \mathbb{C}^{N}$, $w(y) \in \mathbb{C}^{N}$ and $p(y) \in \mathbb{C}^{N}$ are streamwise, wall-normal and spanwise velocities, and pressure, respectively.
Also, $N$ is the number of collocation points in $y$ to evaluate the system, $\nabla \in \mathbb{C}^{3N \times N}$ is the discrete gradient operator and $E(\kappa_x, \kappa_z) \in \mathbb{C}^{4N \times 4N}$, $A(Re, \kappa_x, \kappa_z) \in \mathbb{C}^{4N \times 4N}$, $B(\kappa_x, \kappa_z) \in \mathbb{C}^{4N \times 3N}$ and $C(\kappa_x, \kappa_z) \in \mathbb{C}^{9N \times 4N}$ are the matrix operators.
Readers are referred to the work in \cite{mckeon2017engine} for details on the construction of matrix operators.
It is important to note that $\Delta$ for this system has a repeated complex full-block structure that results from the approximate modeling of the quadratic convective nonlinearity as,
\begin{equation}
    f(y) = \left[\begin{smallmatrix}
        -u_\xi^{\mathrm{T}} & 0 & 0 \\
        0 & -u_\xi^{\mathrm{T}} & 0 \\
        0 & 0 & -u_\xi^{\mathrm{T}} \\
    \end{smallmatrix}\right] \left[\begin{smallmatrix}
        \nabla u \\
        \nabla v \\
        \nabla w \\
    \end{smallmatrix}\right] = (I_3 \otimes -u_\xi^{\mathrm{T}}) \zeta(y)
\end{equation}
where $f(y) \in \mathbb{C}^{3N}$ is the forcing signal and $u_\xi \in \mathbb{C}^{3N \times N}$ is the velocity gain matrix.
Thus, the last row of equations in \eqref{eq:dyn_eq} describes the nonlinear forcing with $\Delta = I_3 \otimes -u_\xi^{\mathrm{T}}$ as the uncertainty matrix.
Further details are given in \cite{liu21} about the $\Delta$ modeling.
The input-output map of the system in \eqref{eq:dyn_eq} is given by,
\begin{equation}
    H(y;Re, \omega,\kappa_x,\kappa_z) = C(\mathrm{i}\omega E - A)^{-1}B,
    \label{eq:full-rank-eq}
\end{equation}
where $\omega$ is the temporal frequency.
$H(y;Re, \omega,\kappa_x,\kappa_z)$ in \eqref{eq:full-rank-eq} is, in general, not a rank-one matrix. However, for demonstration of our method, we will approximate $H(y;Re, \omega,\kappa_x,\kappa_z)$ as a rank-one input-output operator at each of the temporal frequencies $\omega$ for a fixed $Re$, $\kappa_x$ and $\kappa_z$---as is commonly done for such analyses~\cite{mckeon2017engine}:
\begin{equation}
    M_{\omega_i} = \overline{\sigma}_i a_{1_i} b_{1_i}^{\mathrm{H}}, \,\, i = 1, \ldots,N_\omega
\end{equation}
where $N_\omega$ are the total number of frequency points, $\overline{\sigma}_i \in \mathbb{R}_{\geq 0}$ is the maximum singular value of a matrix, and $a_{1_i} \in \mathbb{C}^{9N}$ and $b_{1_i} \in \mathbb{C}^{3N}$ are the left and right unitary vectors associated with $\overline{\sigma}_i$, respectively.
Then, the rank-one SSV is given by $\mu_{\max} = \max_i \,\mu(M_{\omega_i})$, where $\mu(M_{\omega_i})$ is computed using \eqref{eq:mu_M}.
\subsection{Numerical Implementation}
We will compute $\mu_{\max}$ on an $N_{\kappa} \times N_{\kappa} \times N_\omega$ grid of space and temporal frequencies.
The spatial frequencies (wavenumbers) $\kappa_x$ and $\kappa_z$ are both defined on a log-spaced grid of $N_\kappa = 50$ points in the interval $[10^{-1.45}, 10^{2.55}]$. 
This grid is denoted $G_\kappa$. The temporal frequency $\omega$ is defined on a grid $G_\omega:=\{ c_p G_\kappa\}$, where $c_p$ is the wave speed, i.e., speed of the moving base flow (see \cite{mckeon2017engine} for details).  
Wave speeds are chosen as $c_p \in \{5, 10, 15, 18, 22 \}$ resulting in $N_\omega = 250$ points in the temporal frequency grid.
Additionally, we will fix $Re = 180$ and $N = 60$ for all computations and use MATLAB's \texttt{parfor} command to loop over temporal frequencies.
\begin{figure}[!ht]
\begin{center}
\subfigure[Upper-Bound of $\mu_{\max}$]{\includegraphics[width = 0.425\textwidth]{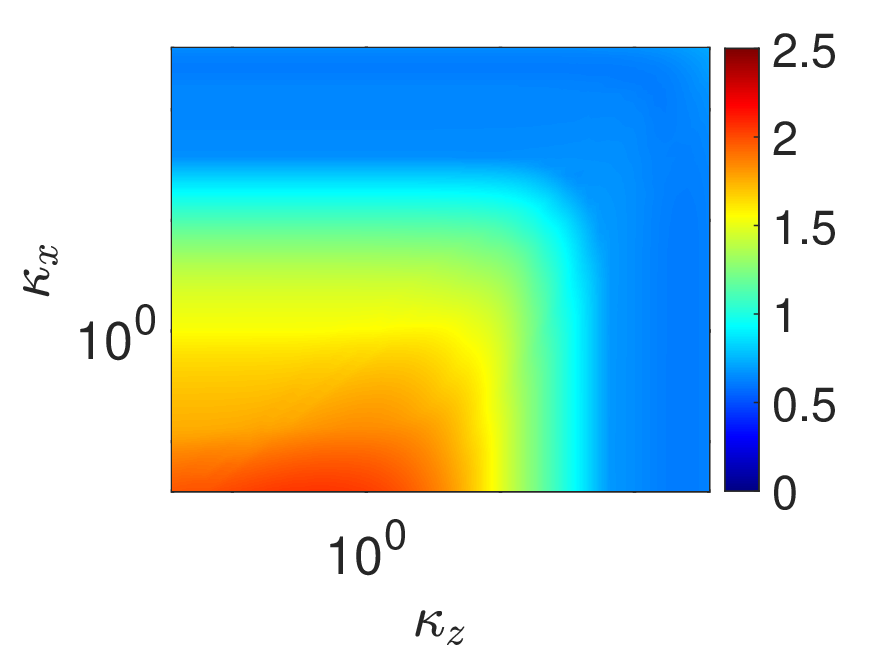} \label{fig:ssv_ub_moc}}
\subfigure[Exact Rank-One $\mu_{\max}$]{\includegraphics[width = 0.425\textwidth]{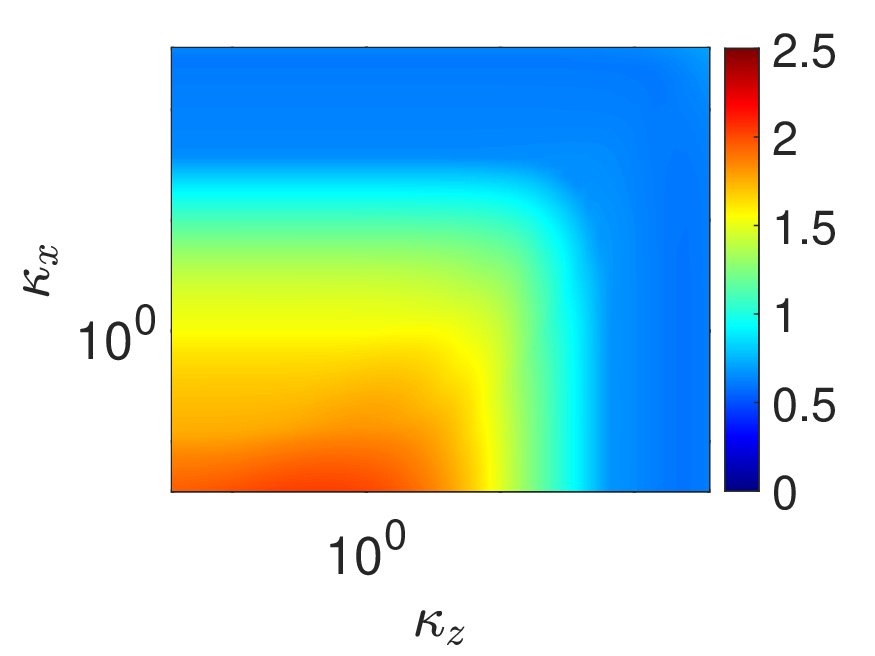} \label{fig:ssv_rank_1_mu}}
\end{center}
\caption{The plots depict the $log_{10}$ values of the upper of $\mu_{\max}$ and $\mu_{\max}$.
We see that $\mu_{\max}$ solutions are similar to the the upper-bounds of $\mu_{\max}$.
The lower-bounds of $\mu_{\max}$ (not shown here) are ``identical" to the $\mu_{\max}$ solutions, i.e., within $1\%$ of each other.}
\label{fig:SSV_rank_1}
\end{figure}
\begin{figure}[!ht]
\begin{center}
\includegraphics[width = 0.52\textwidth]{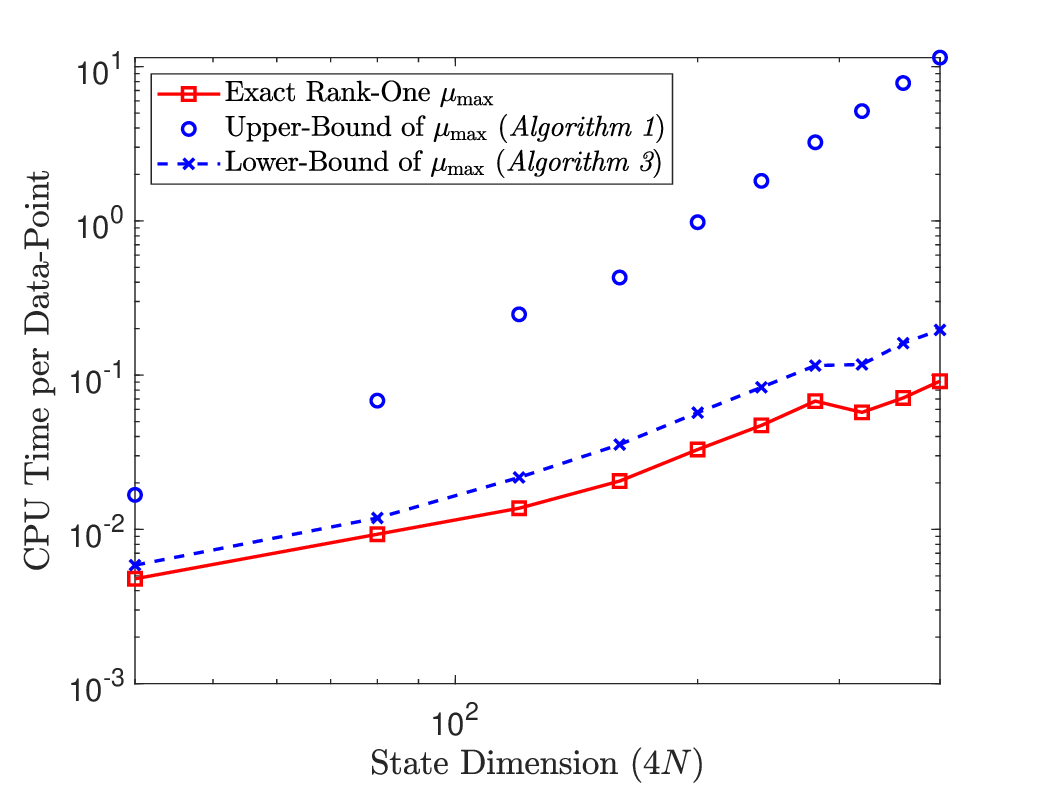}
\end{center}
\caption{The plots show the computational run time for $\mu_{\max}$, and upper and lower-bound calculations of $\mu_{\max}$.}
 \label{fig:runtime}
\end{figure}
\subsection{Discussion}
We can see in figure \ref{fig:SSV_rank_1} that $\mu_{\max}$ values are qualitatively and quantitatively similar (within $5 \%$) to the upper-bounds of $\mu_{\max}$ obtained from \emph{Algorithm 1} in \cite{Mushtaq2023Algorithm}.
In fact, $\mu_{\max}$ values are ``identical" to the lower-bound values of $\mu_{\max}$ (not shown here), i.e., values match up to $1\%$.
Thus, the algorithms converge to the optimal solutions obtained from our method. 

Furthermore, computing $\mu_{\max}$ is relatively fast as compared to obtaining its bounds (see figure \ref{fig:runtime}).
Each point on the plot in figure \ref{fig:runtime} represents the average\footnote{The CPU times are averaged over $10$ data-points. We used an ASUS ROG M15 laptop with Intel 2.6 GHz i7-10750H CPU with 6 cores, 16 GB RAM, and an RTX 2070 Max-Q GPU for run time computations.} CPU time for a single data-point $(\omega,\kappa_x,\kappa_z)$ at each of the state dimensions.
All computational times include CPU time for SVD of $H$ to obtain a rank-one approximation.
From the plot in figure \ref{fig:runtime}, the upper-bound and lower-bound solutions have a time complexity of $\mathcal{O}(N^{2.83})$ and $\mathcal{O}(N^{1.525})$, respectively.
Meanwhile, computing $\mu_{\max}$ from our method has a time complexity of $\mathcal{O}(N^{1.28})$.
%
%
%
%
\section{Conclusion}
%
This work presents an exact solution of SSV for rank-one complex matrices with repeated, complex full-block uncertainties. 
The solution obtained from this method generalizes previous exact solutions for the repeated complex scalar and/or non-repeated complex full-block uncertainties \cite{chen1994structuredpart1,young1994rank}.
We illustrated the proposed method on a turbulent channel flow model. 
In future work, we would like to explore similar arguments to the ones presented here for rank-one complex matrices to compute SSV for general (not necessarily rank-one) complex matrices, especially when $\Delta \in \mathbf{\Delta}$.
\section{Acknowledgements}
This material is based upon work supported by the ARO under grant number W911NF-20-1-0156. 
%
MSH acknowledges support from the AFOSR under award number FA 9550-19-1-0034, 
the NSF under grant
number CBET-1943988 and ONR under award number N000140-22-1-2029.
\bibliographystyle{IEEEtran}
\bibliography{ref}
\end{document}